\documentclass[runningheads]{llncs}
\usepackage[T1]{fontenc}

\usepackage{amsmath}
\usepackage{amssymb}
\usepackage{amsthm}
\usepackage{mathtools}
\usepackage{etoolbox}

\usepackage{xcolor}
\usepackage{url}
\usepackage{wrapfig}
\usepackage{multirow}
\usepackage{graphicx}

\usepackage{algorithmic}
\usepackage{algorithm}

\newtheorem{assumption}{Assumption}

\def\nN{\mathcal{N}}
\def\fJ{J}

\def\fF{\mathcal{F}}

\def\R{\mathbb{R}}

\def\N{\mathbb{N}}

\def\G{\mathcal{G}}
\def\K{\mathbf{K}}
\def\L{\mathbf{L}}
\def\S{\mathbf{S}}

\newcommand{\eE}{\mathbb{E}}

\DeclareMathOperator*{\argmin}{arg\,min}
\DeclareFontFamily{U}{matha}{\hyphenchar\font45}
\DeclareFontShape{U}{matha}{m}{n}{
<-6> matha5 <6-7> matha6 <7-8> matha7
<8-9> matha8 <9-10> matha9
<10-12> matha10 <12-> matha12
}{}
\DeclareSymbolFont{matha}{U}{matha}{m}{n}
\DeclareFontFamily{U}{mathx}{\hyphenchar\font45}
\DeclareFontShape{U}{mathx}{m}{n}{
<-6> mathx5 <6-7> mathx6 <7-8> mathx7
<8-9> mathx8 <9-10> mathx9
<10-12> mathx10 <12-> mathx12
}{}
\DeclareSymbolFont{mathx}{U}{mathx}{m}{n}
\DeclareMathDelimiter{\vvvert} {0}{matha}{"7E}{mathx}{"17}
\DeclarePairedDelimiterX{\normiii}[1]
{\vvvert}
{\vvvert}
{\ifblank{#1}{\:\cdot\:}{#1}}

\begin{document}

\title{Equivariant Denoisers for Image Restoration}
\author{Marien Renaud\thanks{corresponding author : marien.renaud@math.u-bordeaux.fr}\inst{1,2} \and
Arthur Leclaire\inst{1} \and
Nicolas Papadakis\inst{2}}

\institute{LTCI, T\'el\'ecom Paris, IP Paris \and Univ. Bordeaux, CNRS, INRIA, Bordeaux INP, IMB, UMR 5251, F-33400 Talence, France, }

\maketitle

\begin{abstract}
One key ingredient of image restoration is to define a realistic prior on clean images to complete the missing information in the observation.
State-of-the-art restoration methods rely on a neural network to encode this prior. 
Moreover, typical image distributions are invariant to some set of transformations, such as rotations or flips.
However, most deep architectures are not designed to represent an invariant image distribution.
Recent works have proposed to overcome this difficulty by including equivariance properties within a Plug-and-Play paradigm.
In this work, we propose a unified framework named Equivariant Regularization by Denoising (ERED) based on equivariant denoisers and stochastic optimization. We analyze the convergence of this algorithm and discuss its practical benefit.

\end{abstract}

\section{Introduction}
Image restoration aims at recovering a proper image $x \in \R^d$ from a degraded observation $y \in \R^m$. 
A model to obtain $y$ from $x$ can be defined as $y \sim \mathcal{N}(\mathcal{A}(x))$,
where $\mathcal{A}: \R^d \to \R^m$ is a deterministic operation on $x$ and $\mathcal{N}$ is a noise distribution that corresponds to a known physical degradation. Typically, linear degradations with Gaussian noise can be written as $y = A x + n$, with $A \in \R^{m \times d}$ and $n \sim \nN(0, \sigma_y^2 I_m)$.

\textbf{The regularization choice} 
In this context, the restoration task can be reformulated as a variational problem
\begin{align}\label{eq:ideal_opt_pblm}
    \argmin_{x \in \R^d} \fF(x) = f(x) + \lambda r(x),
\end{align}
where the data-fidelity term $f = - \log{p(y|\cdot)}$ is the log-likelihood representing the degradation, while $r = - \log p$ is the regularization term encoding the prior distribution $p$, i.e. the model for clean images. 
The choice of the regularization $r$ is crucial to complete the missing information in the observation $y$, e.g. by enforcing piecewise constant images~\cite{rudin1992nonlinear} or wavelet sparsity~\cite{mallat1999wavelet}. To regularize the problem,  recent generic methods rely on the use of an external denoiser within the Plug-and-Play (PnP) framework~\cite{venkatakrishnan2013plug}. 
This denoiser $D_{\sigma}$, parametrized by a noise level $\sigma$, is generally a pre-trained deep neural network~\cite{zhang2021plug}.

\textbf{Image invariances}  Existing trained denoisers do not encode invariances on the underlying prior distribution. In fact, image distributions can be expected to be invariant to some transformations such as rotations or flips. Recent works develop Plug-and-Play restoration methods that try to enforce invariances on the underlying prior in order to improve restoration~\cite{celledoni2021equivariant,chen2021equivariant,fu2023rotationequivariantproximaloperator,terris2024equivariant,tachella2023equivariantbootstrappinguncertaintyquantification,mbakam2024empirical}.

\textbf{Stochastic version of PnP} Another way to increase the performance of PnP focuses on the development of stochastic versions of PnP~\cite{laumont2023maximum}. Stochasticity helps to reduce the computational cost~\cite{sun2019block,tang2020fast}.
Moreover, stochastic algorithms might compute better approximate solutions of problem~\eqref{eq:ideal_opt_pblm}~\cite{hu2024stochasticdeeprestorationpriors,renaud2024plug}.

\textbf{Contributions} \textbf{(1)} In this paper, we propose the unified  $\pi$-equivariance formalism for the PnP framework, which generalizes so-called equivariant PnP~\cite{terris2024equivariant} and other stochastic versions of PnP~\cite{hu2024stochasticdeeprestorationpriors,renaud2024plug}. \textbf{(2)} We give theoretical insights on the convergence guarantees of the proposed Equivariant Regularization by Denoising (ERED, Algorithm~\ref{alg:ERED}). \textbf{(3)} We study the convergence of the ERED critical points (Problem~\ref{eq:equivariant_opt_pblm}) when the denoiser parameter $\sigma$ goes to zero, in the case of a $\pi$-equivariant prior $p$. \textbf{(4)} We provide numerical experiments and comparisons for image restoration tasks, and discuss the practical benefits of equivariant approaches.

\section{Background on Regularization by Denoising}
In order to solve problem~\eqref{eq:ideal_opt_pblm}, when $f$ and $r$ are differentiable, one can use a gradient descent algorithm. However, the gradient of $r$, i.e. the score of the prior distribution $p$ of clean images $s := \nabla \log p = - \nabla r$ is unknown. The authors of~\cite{romano2017little} proposed to make the following approximation 
\begin{align}
    \nabla r = - \nabla \log p \approx \nabla r_{\sigma} := - \nabla \log p_{\sigma},
\end{align} 
where $p_{\sigma} = \nN_{\sigma} \ast p$ is the convolution of $p$ with the Gaussian $\nN_{\sigma}$ with $0$-mean and $\sigma^2 I_d$-covariance matrix. This is motivated by the Tweedie formula~\cite{efron2011tweedie}, which makes $\nabla \log p_{\sigma}$ tractable
\begin{align}\label{eq:tweedie_formula}
    -\nabla \log p_{\sigma}(x) = \frac{1}{\sigma^2} \left( x - D_{\sigma}^\ast(x) \right),
\end{align}
where $D_{\sigma}^\star$ is the Minimum Mean Square Error (MMSE) denoiser defined by $D_{\sigma}^\star(\tilde x) := \eE [x | \tilde x ] = \int_{\R^d}{x p(x | \tilde x) dx}$,
for $\tilde x = x + \epsilon \text{ with } x \sim p(x), \epsilon \sim \nN(0, \sigma^2 I_d)$.
This approximation leads to the Regularization By Denoising (RED) iterations:
\begin{align}
    x_{k+1} &=  x_k - \delta \nabla f(x_k) - \delta \frac{\lambda}{\sigma^2} \left( x_k - D_{\sigma}(x_k)\right),
\end{align}
where $D_\sigma$ is a denoiser that is designed to approximate $D_{\sigma}^{\star}$.

The performance of RED can be improved by slightly modifying the algorithm with stochastic schemes~\cite{terris2024equivariant,renaud2024plug}, which incorporate invariance properties and enhance details in the restoration. In the next section, we propose a unified framework to generalize these approaches.

\section{$\pi$-Equivariant Regularization by Denoising}
We now  introduce an extension of RED that is deduced from a notion of invariance, named $\pi$-equivariance, on the underlying prior. \vspace*{0.3cm} 

\noindent {\em Notations}
We denote by $g : \R^d \to \R^d$ a differentiable transformation of $\R^d$; by $\G$ a measurable set of transformations of $\R^d$, and by $G \sim \pi$ a random variable of law $\pi$ on $\G$.

\subsection{$\pi$-equivariant image distributions}
\begin{definition}[Invariance]\label{def:invariance}
    A density $p$ on $\R^d$ is said to be invariant to a set of transformations $\G$ if $\forall g \in \G$, $p = p \circ g$ a.e.
\end{definition}

As shown in~\cite{lenc2015understanding}, natural images densities tend to be invariant to some set of transformations such as rotations or flips have been studied in.

\begin{definition}[$\pi$-equivariance]\label{def:equivariance}
    A density $p$ on $\R^d$ is said to be $\pi$-equivariant if 
    $\eE_{G \sim \pi}[|\log(p \circ G)|] < \infty$
    and 
    $\log p = \eE_{G \sim \pi}\left[ \log(p \circ G) \right]$.
\end{definition}

Definition~\ref{def:equivariance} relaxes the notion of invariance for a density in the following sense. If a density $p$ is invariant to each $g \in \G$, $p$ is $\pi$-equivariant, whatever the distribution $\pi$ on $\G$. 

\begin{remark}[Key identity]\label{remark:score_invariance}
For $p \in \mathcal{C}^1(\R^d, \R_{+}^\ast)$ and $g \in \mathcal{C}^1 (\R^d, \R^d)$, we have
\begin{align}\label{eq:score_composition}
    \nabla \log (p \circ g)(x)
    &= \frac{\nabla (p \circ g)(x)}{(p \circ g)(x)} = \frac{\fJ_g^T(x)\nabla p(g(x))}{(p \circ g)(x)} = \fJ_g^T(x) (\nabla \log p)(g(x)),
\end{align}
with $x \in \R^d$. Thus, if $p$ is $\pi$-equivariant, then $s=-\nabla \log p$, the score of $p$,  verifies the identity $s = \eE_{G \sim \pi}\left( \fJ_G^T (s \circ G) \right)$.

\end{remark}

In the context of equivariant transforms, Remark~\ref{remark:score_invariance} suggests to apply the Jacobian of the transformation $g$ instead of the inverse of the transformation $g^{-1}$ as it is done in existing  works~\cite{mbakam2024empirical,tachella2023equivariantbootstrappinguncertaintyquantification,terris2024equivariant,herbreteau2024normalization}. Here the score~\eqref{eq:score_composition} can be computed for any general differentiable transformation $g$, even if $g^{-1}$ does not exist.

\subsection{Equivariant regularization}
In order to encode the desired equivariance property in the regularization $r$, we introduce the \textbf{Equivariant REgularization by Denoising (ERED)} $ r_{\sigma}^{\pi}$ and the associated score $s_{\sigma}^{\pi}$ respectively defined by
\begin{align}
    r_{\sigma}^{\pi}(x) &:= -\eE_{G \sim \pi} \left( \log (p_{\sigma} \circ G)(x) \right) \label{eq:eq_reg} \\
    s_{\sigma}^{\pi}(x) &:= -\eE_{G \sim \pi} \left(\fJ_G^T(x) (\nabla \log p_{\sigma})(G(x)) \right).\label{eq:equivariant score}
\end{align}

Note that under regularity assumptions on $\pi$ and $p_{\sigma} \circ G$ (e.g. $\mathcal{G}$ finite and $p_{\sigma} \circ G$ differentiable), we get $s_{\sigma}^{\pi} = \nabla r_{\sigma}^{\pi}$.
Thanks to the Tweedie formula~\eqref{eq:tweedie_formula}, $s_{\sigma}^{\pi}$ can be computed with an MMSE denoiser
\begin{align}\label{eq:equiv_score}
    s_{\sigma}^{\pi}(x) &= \eE_{G \sim \pi} \left(\frac{1}{\sigma^2}  \fJ_G^T(x) \left(G(x) - D_{\sigma}^\ast(G(x)) \right) \right).
\end{align}
For a given denoiser $D_{\sigma}$, e.g. a supervised neural network, we can thus introduce the \textit{equivariant denoiser} $\tilde D_{\sigma}$ defined~by
\begin{align}\label{eq:equiv_denoiser}
    \tilde D_{\sigma}(x) = \eE_{G \sim \pi} \left[ \fJ_{G}^T(x) D_{\sigma} \left( G (x) \right) \right].
\end{align}

Since the exact MMSE denoiser $D_{\sigma}^\ast$ is not tractable, we make the following approximation:
\begin{align*}
    s_{\sigma}^{\pi}(x) = \frac{1}{\sigma^2} \left( \eE_{\pi} \left[\fJ_G^T(x) G(x) \right] - \tilde D_{\sigma}^\ast(x)  \right) \approx \frac{1}{\sigma^2} \left( \eE_{\pi} \left[\fJ_{G}^T(x) G(x) \right] - \tilde D_{\sigma}(x) \right).
\end{align*}

No special structure is required on $\G$ or $\pi$, thus making the  ERED framework a generic construction. 
However, in order to ensure that the ERED $ r_{\sigma}^{\pi}$ is indeed $\pi$-equivariant (Definition~\ref{def:equivariance}), more structure is required on $\G$ and $\pi$. 
In Proposition~\ref{prop:haar}, a sufficient condition on $\G$ and $\pi$ is provided for this property to hold.
Before stating this result, let us recall that a compact Hausdorff topological group admits a unique right-invariant probability measure $\pi$, called Haar measure, that satisfies, for any integrable function $\varphi : \G \to \R$ and for any $g \in \G$, $\int_{\G} \varphi(g(x)) d\pi(x) = \int_{\G} \varphi(x) d\pi(x)$~\cite{haar1933massbegriff,neumann1935haarschen}. For a finite group $\G$, the Haar measure is the counting measure.

\begin{proposition}\label{prop:haar}
    If $\G$ is a compact Hausdorff topological group and $\pi$ the associated right-invariant Haar measure, then $r_{\sigma}^{\pi}$ is $\pi$-equivariant.
\end{proposition}

\begin{proof}
For $x \in \R^d$, by the right-invariance of $\pi$, we get
\begin{align*}
    &\eE_{G' \sim \pi}\left( \log(r_{\sigma}^{\pi}(G'(x)) \right)= \int_{\G} \log(r_{\sigma}^{\pi}(G'(x)) d\pi(G') \\
    =& \int_{\G} \log\left( -\int_{\G}  \log (p_{\sigma} (G \circ G'(x)) d\pi(G) \right) d\pi(G') \\
    = &\int_{\G} \log\left( -\int_{\G}  \log (p_{\sigma} (G(x)) d\pi(G) \right) d\pi(G') 
     = \log\left( r_{\sigma}^{\pi}(x) \right).\qedhere
\end{align*}
\end{proof}
Let us discuss the hypothesis on the set of transformations $\G$.
First, $\G$ needs to be a group, i.e. $\forall g, g' \in \G, g^{-1} \circ g' \in \G$, to ensure that the composition $G \circ G'$ is still in $\G$ and that any transformation is invertible.
Moreover, $\G$ needs to be a Hausdorff space (i.e. $\forall g, g' \in G$, there exist two neighbourhoods $U$ and $V$ of  $x$ and $y$, such that $U \cap V = \emptyset$) and is required to be compact to ensure that $\pi$ is a probability measure.
These hypotheses on $\G$ are general and cover in particular any finite discrete $\G$.

\subsection{Examples of equivariant scores}\label{sec:equivariant_ex}
The equivariant formulation generalizes recent works in the literature that we recall here. Moreover, we present new examples of equivariant scores that are included in our general framework.

\paragraph{Finite set of isometries}
The authors of \cite{terris2024equivariant} proposed an equivariant version of PnP for a finite set of linear isometric transformations $\G$ with the uniform distribution on $\G$, i.e. $\forall g \in \G, \pi(g) = \frac{1}{|\G|}$. Since $g$ is a linear isometry, $\fJ_g^T(x) = g^{-1}$.
In this case, the $\pi$-equivariant denoiser~\eqref{eq:equiv_denoiser} and  regularization~\eqref{eq:equiv_score} are respectively defined by
\begin{align}
    \tilde D_{\sigma}(x) = \frac{1}{|\G|} \sum_{g \in \G}{ g^{-1}\left[ D_{\sigma} \left( g(x) \right) \right] }, ~~ s_{\sigma}^{\pi}(x) \approx \frac{1}{\sigma^2} \left(x - \tilde D_{\sigma}(x) \right).
\end{align}

\paragraph{Infinite set of isometries}
To generalize the previous example, we can propose an infinite set $\G$ of isometries, e.g. sub-pixel rotations for which $\pi$ can be seen as the angular distribution.
In this case, $\forall g \in \G, \forall x \in \R^d,\  \fJ_g^T(x) = g^{-1}$ and then
\begin{align}
    \tilde D_{\sigma}(x) = \eE_{G \sim \pi}\left( G^{-1}\left[ D_{\sigma} \left( G(x) \right) \right] \right),~~ s_{\sigma}^{\pi}(x) \approx \frac{1}{\sigma^2} \left(x - \tilde D_{\sigma}(x) \right).
\end{align}

\paragraph{Noising-denoising} The stochastic denoising regularization (SNORE) proposed in~\cite{renaud2024plug} consists in noising the image before denoising it. It can be interpreted as a $\pi$-equivariant denoiser for the set of translations $g_{z}(x) = x + \sigma z$, for $x, z \in \R^d$ and $\sigma > 0$ the noise level of $D_{\sigma}$, with the multivariable distribution, i.e. ${\forall z \in \R^d}$, $\pi(g_{z}) = \nN(z; 0, \sigma^2 I_d)$. The Jacobian of the translation is $\fJ_{g_z}(x) = I_d$. With this set of transformations, the $\pi$-equivariant  denoiser~\eqref{eq:equiv_denoiser} and regularization~\eqref{eq:equiv_score} can be expressed as 
\begin{align}
    \tilde D_{\sigma}(x) &= \eE_{z \sim \N(0, \sigma I_d)} \left[ D_{\sigma} \left( x + \sigma z \right) \right] \label{eq:noising_denoising_denoiser} \\
    s_{\sigma}^{\pi}(x) &\approx \frac{1}{\sigma^2} \left( \eE_{\pi} \left[x + \sigma z \right] - \tilde D_{\sigma}(x) \right) = \frac{1}{\sigma^2} \left( x - \tilde D_{\sigma}(x) \right).
\end{align}

\section{ERED algorithm: definition and analysis}

In this section, we define a generic $\pi$-equivariant PnP algorithm and demonstrate its convergence.
Instead of solving Problem~\eqref{eq:ideal_opt_pblm}, we will tackle 
\begin{align}\label{eq:equivariant_opt_pblm}
    \argmin_{x \in \R^d} \mathcal{F}_{\sigma}^{\pi}(x) := f(x) + \lambda r_{\sigma}^{\pi}(x),
\end{align}
with the equivariant regularization $r_{\sigma}^{\pi}$ defined in relation~\eqref{eq:eq_reg}.
We now introduce the equivariant Regularization by Denoising (ERED) algorithm (Algorithm~\ref{alg:ERED}), which is a biased stochastic gradient descent to solve Problem~\eqref{eq:equivariant_opt_pblm}.

\begin{algorithm}
\caption{ERED}\label{alg:ERED}
\begin{algorithmic}[1]
\STATE \textbf{Parameters:} $x_0 \in \R^d$, $\sigma > 0$, $\lambda > 0$, $\delta > 0$, $N \in \N$
\STATE \textbf{Input:} degraded image $y$
\STATE \textbf{Output:} restored image $x_{N}$
\FOR{$k = 0, 1, \dots, N-1$}
    \STATE Sample $G \sim \pi$
    \STATE $x_{k+1} = x_k - \delta \nabla f(x_k) - \frac{\delta \lambda}{\sigma^2} \fJ_{G}^T(x_k) \left( G(x_k) - D_{\sigma}(G(x_k))\right)$
\ENDFOR
\end{algorithmic}
\end{algorithm}

\begin{remark}
    If the denoiser $D_{\sigma}$ is $L$-Lipschitz then the equivariant denoiser $\tilde D_{\sigma}$ defined in relation~\eqref{eq:equiv_denoiser} is $L$-Lipschitz. 
    Previous works have shown the link between the Lipschitz constant of the denoiser and the convergence of deterministic PnP algorithms~\cite{ryu2019plug,hurault2022proximal,wei2024learning}. However, no clue indicates that this property still holds for stochastic PnP algorithms, such as Algorithm~\ref{alg:ERED}. Thus, in this paper, we provide a new analysis to demonstrate the convergence of   Algorithm~\ref{alg:ERED}.
\end{remark}

\subsection{Unbiased Convergence analysis}

In this section, we prove the convergence of the ERED (Algorithm~\ref{alg:ERED}) run with the exact MMSE denoiser~$D_{\sigma}^\star$. With this denoiser, thanks to Tweedie formula, the iterations are computed by
\begin{align}\label{eq:theoretical_process}
    x_{k+1} = x_k - \delta_k \nabla f(x_k) - \lambda \delta_k \fJ_{G}^T(x_k) \nabla \log p_{\sigma}(G (x_k)),
\end{align}
with $G \sim \pi$ and $(\delta_k)_{k \in \N} \in {\left(\R^+\right)}^{\N}$ a non-increasing sequence of step-sizes.

\begin{assumption}
    \itshape \textbf{(a)} \label{ass:step_size_decreas} The step-size decreases to zero but not too fast: $\sum_{k = 0}^{+\infty}{\delta_k} = + \infty$ and $\sum_{k = 0}^{+\infty}{\delta_k^2} < + \infty$.
    
    \itshape \textbf{(b)} \label{ass:data_fidelity_reg} The data-fidelity term $f : x \in \R^d \mapsto f(x) \in \R$ is $\mathcal{C}^{\infty}$.
    
    \itshape \textbf{(c)} \label{ass:prior_score_approx} The noisy prior score is sub-polynomial, i.e. there exist $B\in \R^+$, $\beta \in \R$ and $n_1 \in \N$ such that $\forall \sigma > 0$, $\forall x \in \R^d$, $\|\nabla \log p_{\sigma}(x)\| \le B \sigma^{\beta} (1 + \|x\|^{n_1})$.
\end{assumption}

Assumption~\ref{ass:step_size_decreas}(a) is standard in stochastic gradient descent analysis. It suggests a choice of the step-size rule to ensure convergence, for instance $\delta_k = \frac{\delta}{k^\alpha}$ with $\alpha \in ]\frac{1}{2},1]$.
Assumption~\ref{ass:data_fidelity_reg}(b) is typically verified for a linear degradation with additive Gaussian noise, i.e. $f(x) = \frac{1}{\sigma_y^2}\|y-Ax\|^2$. It ensures that the objective function of Problem~\eqref{eq:equivariant_opt_pblm} is $\mathcal{C}^{\infty}$. Under the so-called manifold hypothesis, i.e. $p$ is supported on a compact, it is shown in~\cite{debortoli2023convergence} that  Assumption~\ref{ass:prior_score_approx}(c) is verified with $n_1 = 1$ and $\beta = -2$.

\begin{assumption}\label{ass:finite_moment}
    \textbf{(a)} The random variable $\fJ_G$ has a uniform finite moment, i.e. $\exists \epsilon > 0, M_{2+\epsilon} \ge 0$ such that $\forall x \in \R^d, \eE_{G \sim \pi}(\vvvert \fJ_G(x) \vvvert^{2+\epsilon}) \le M_{2+\epsilon} < +\infty$, with $\vvvert \cdot \vvvert$ the operator norm defined for $A \in \R^{d\times d}$ by $\vvvert A \vvvert = \sup_{\|x\|=1} \frac{\|Ax\|}{\|x\|}$.

\noindent
    \textbf{(b)} The transformation has bounded moments on any compact, i.e. $\forall\K\subset \R^d$ compact, $\forall m \in \N$, $\exists C_{\K, m}< + \infty$ such that $\forall x \in \K, \eE_{G \sim \pi}(\|  G(x) \|^m) \le C_{\K, m}$.
\end{assumption}
With Assumption~\ref{ass:finite_moment}, the behavior of the algorithm is controlled on each compact set. 
This assumption is verified for all examples presented in Section~\ref{sec:equivariant_ex}.

We now define $\mathbf{S}_{\sigma} = \{x \in \R^d | \nabla \fF_{\sigma}^{\pi}(x) = 0\}$, the set of critical points of $\fF_{\sigma}^{\pi}$ in $\R^d$, and $\Lambda_{\K}$, the set of random seeds for which the iterates of the algorithm are bounded in the compact $K$, by
\begin{equation*}
    \Lambda_{\K} = \bigcap_{k \in \N}{\{x_k \in \K\}}.
\end{equation*}
We finally denote the distance of a point to a set by $d(x, \mathbf{S}) = \inf_{y \in \mathbf{S}}{\|x - y \|}$, with $x \in \R^d$ and $\mathbf{S} \subset \R^d$.
The restriction to realizations of the algorithm bounded in $\Lambda_{K}$ will be referred to as the \textit{boundedness assumption}.

\begin{proposition}\label{prop:convergence_unbiased}
    Let $(x_k)_{k \in \N}$ be the iterates generated by Algorithm~\ref{alg:ERED} with the exact MMSE Denoiser $D^*_\sigma$. Then, under Assumptions~\ref{ass:step_size_decreas}-\ref{ass:finite_moment}, we have almost surely on $\Lambda_{\K}$
    \begin{align}
        &\lim_{k \to + \infty}{d(x_k, \mathbf{S_{\sigma}})} = 0,\label{eq1a}\\
        &\lim_{k \to + \infty}{\|\nabla \fF_{\sigma}^{\pi}(x_k)\|} = 0,\label{eq1b}
    \end{align}
and $(\fF_{\sigma}^{\pi}(x_k))_{k \in \N}$ converges to a value of $\fF_{\sigma}^{\pi}(\mathbf{S_{\sigma}})$.
\end{proposition}

\begin{proof}
The proof is obtained by applying  \cite[Theorem 2.1, (ii)] {tadic2017asymptotic}, which is recalled in Appendix~\ref{app:thm} (Theorem~\ref{theorem:tadic}) for completeness. To do so, we have to verify the different assumptions (Assumption~\ref{ass:tadic} in Appendix~\ref{app:thm}) of this theorem. 
First, Assumption~\ref{ass:tadic}(a) is verified by Assumption~\ref{ass:step_size_decreas}(a). 
Next $p_{\sigma}$ is $\mathcal{C}^{\infty}$ by convolution with a Gaussian. Then $\log p_{\sigma}$ is also $\mathcal{C}^{\infty}$~\cite{laumont2023maximum} and so is $r_{\sigma}^{\pi}$ defined in relation~\eqref{eq:eq_reg}. By Assumption~\ref{ass:data_fidelity_reg}(b), $\mathcal{F}_{\sigma}^{\pi} = f + \lambda r_{\sigma}^{\pi}$ is $\mathcal{C}^{\infty}$. So Assumption~\ref{ass:tadic}(c) is verified.

Now, we are going to prove that Assumption~\ref{ass:tadic}(b) is verified, i.e. the noise fluctuation can be controlled.
To that end, we define 
\begin{align*}
\xi_k &= \nabla f(x_k) + \lambda \fJ_{G}^T(x_k) \nabla \log p_{\sigma}(G (x_k)) - \nabla \fF_{\sigma}^{\pi} (x_k)\\& = \lambda \left( \fJ_{G}^T(x_k) \nabla \log p_{\sigma}(G (x_k)) - \eE_{G \sim \pi}\left(\fJ_{G}^T(x_k) \nabla \log p_{\sigma}(G(x_k)) \right) \right),\end{align*} 
with $\lambda >0, G \sim \pi$. By definition we have $\eE(\xi_k) = 0$ and, from~\eqref{eq:theoretical_process}, we get
\begin{align*}
    x_{k+1} = x_k - \delta_k (\nabla \fF_{\sigma}^{\pi}(x_k) + \xi_k).
\end{align*}

\begin{lemma}\label{lemma:bounded_variance_1}
    Under Assumptions~\ref{ass:prior_score_approx}-\ref{ass:finite_moment}, almost surely on $\Lambda_{\K}$, there exists $C > 0$, such that $\forall k \in \N$, $\eE(\|\xi_k\|^2) \le C$.
\end{lemma}

Lemma~\ref{lemma:bounded_variance_1} is demonstrated in Appendix~\ref{sec:proof_bounded_variance_1}. By using Lemma~\ref{lemma:bounded_variance_1}, we get $\sum_{k \in \N}{\delta_k^2 \eE(\|\xi_k\|^2)} \le C \sum_{k \in \N}{\delta_k^2} < + \infty$, by Assumption~\ref{ass:step_size_decreas}. Then, we deduce from the Doob inequality (which holds because $\sum_{k=n}^l\xi_k$ is a martingale) that
\begin{align}\label{eq:doob_inequality}
    \eE\left(\sup_{n \le l \le m}{\| \sum_{k=n}^l{\delta_k \xi_k} \|^2}\right) \le 4 \sum_{k = n}^m {\delta_k^2 \eE\left(\|\xi_k\|^2\right)}.
\end{align}
By the monotone convergence theorem, this implies
\begin{equation*}
    \eE\left(\sup_{n \le l}{\| \sum_{k=n}^l{\delta_k \xi_k} \|^2}\right) \le 4 \sum_{k = n}^{\infty} {\delta_k^2 \eE\left(\|\xi_k\|^2\right)} 
    \leq 4 C \sum_{k = n}^{\infty} \delta_k^2 .
\end{equation*}
Thus the sequence $(\sup_{n \le l}{\| \sum_{k=n}^l{\delta_k \xi_k} \|^2})_n$ tends to zero in $L^1$ and also almost surely (because it is non-increasing). The square function being non-decreasing, it implies that $(\sup_{n \le l}{\| \sum_{k=n}^l{\delta_k \xi_k} \|})_n$ tends to zero almost surely.
The process~\eqref{eq:theoretical_process} thus verifies Assumption~\ref{ass:tadic}(b) almost surely. 
We can apply Theorem~\ref{theorem:tadic}, which concludes the proof.
\end{proof}

\subsection{Biased convergence analysis}
In this section, we analyse the convergence of the ERED algorithm (Algorithm~\ref{alg:ERED}) run with a realistic denoiser $D_{\sigma} \neq D_{\sigma}^\ast$. In this case, ERED is a biased stochastic gradient descent for solving Problem~\eqref{eq:equivariant_opt_pblm}. At each iteration, the  algorithm writes
\begin{align}\label{eq:algo_biased}
    x_{k+1} = x_k - \delta_k \nabla f(x_k) - \frac{\delta_k \lambda}{\sigma^2} \fJ_{G}^T(x_k) \left( G(x_k) - D_{\sigma}(G(x_k))\right),
\end{align}
with $G \sim \pi$. Defining the gradient estimator $$\xi_{k} = \nabla f(x_k) + \frac{\lambda}{\sigma^2} \fJ_{G}^T(x_k) \left( G(x_k) - D_{\sigma}(G(x_k))\right) - \nabla \fF_{\sigma}^{\pi}(x_k),$$ the algorithm~\eqref{eq:algo_biased} can be reformulated as 
\begin{align}
    x_{k+1} = x_k - \delta_k \left( \nabla \fF_{\sigma}^{\pi}(x_k) + \xi_k \right).
\end{align}

\begin{assumption}\label{ass:denoiser_sub_polynomial}
    The realistic denoiser $D_{\sigma}$ is sub-polynomial, i.e. $\exists C > 0$ and $n_2 \in N$ such that $\forall x \in \R^d$, $\|D_{\sigma}(x)\| \le C (1 + \|x\|^{n_2})$.
\end{assumption}

Assumption~\ref{ass:denoiser_sub_polynomial} is similar to Assumption~\ref{ass:prior_score_approx}(c) but it applies on the denoiser instead of the score of the underlying prior distribution. Notice that in this case the constant $C$ might depend on the noise level $\sigma$. 
As an example, a bounded denoiser~\cite{chan2016plug} verifies Assumption~\ref{ass:denoiser_sub_polynomial} with $n_2 = 1$.

\begin{assumption}
    \label{ass:g_bounded}
    For every compact $\K$, there exists $C_{\K}$, such that 
    $\forall x \in \K$, $\forall g \in \G, \|g(x)\| \le C_{\K}$. 
\end{assumption}
Assumption~\ref{ass:g_bounded} is verified as soon as $(g, x) \in \G \times \R^d \to g(x) \in \R^d$ is continuous and $\G$ compact. It is verified in particular for $\G$ being a finite set of isometries.

\begin{proposition}\label{prop:convergence_biaised}
    Let  $(x_k)_{k \in \N}$ be the sequence provided by ERED (Algorithm~\ref{alg:ERED}) with an inexact denoiser $D_{\sigma}$. Then,  under Assumptions~\ref{ass:step_size_decreas}-\ref{ass:denoiser_sub_polynomial},  there exists $M_{\K}$ such that, almost surely on $\Lambda_{\K}$:
    \begin{align}
        \limsup_{k \to \infty} \|\nabla \fF_{\sigma}^{\pi}(x_k)\| &\le M_{\K} \eta^{\frac{1}{2}} \label{eq1}\\
        \limsup_{k \to \infty} \fF_{\sigma}^{\pi}(x_k) - \liminf_{k \to \infty} \fF_{\sigma}^{\pi}(x_k) &\le M_{\K} \eta,\label{eq2}
    \end{align}
    with the asymptotic bias $\eta = \limsup_{k \to \infty} \|\eE(\xi_k)\|$.

    Moreover, under Assumption~\ref{ass:g_bounded}, we have
    \begin{align}\label{ineq}
        \eta \le \frac{\lambda}{\sigma^2} \sup_{x \in \K}\eE \left( \vvvert \fJ_G(x)\vvvert\right) \|D_{\sigma} - D_{\sigma}^\ast\|_{\infty, \L},
    \end{align}
    with $\L = \mathcal{B}(0, C_{\K})$, where $C_{\K}$ is introduced in Assumption~\ref{ass:g_bounded}.
\end{proposition}

\begin{proof}
First the bias is denoted by $\eta_k = \eE(\xi_k)$ and the noise by ${\gamma_k = \xi_k - \eE(\xi_k)}$. So we have $\xi_k = \gamma_k + \eta_k$ and $\eE(\gamma_k) = 0$.
We apply again~Theorem~\ref{theorem:tadic}. Assumptions~\ref{ass:tadic}(a)-\ref{ass:tadic}(c) are verified thanks to Assumptions~\ref{ass:step_size_decreas}(a) and~\ref{ass:data_fidelity_reg}(b).

\begin{lemma}\label{lemma:bounded_variance_2}
    Under Assumptions~\ref{ass:finite_moment}-\ref{ass:denoiser_sub_polynomial}, almost surely on $\Lambda_{\K}$, there exists $C_2 > 0$, such that $\forall k \in \N$, $\eE(\|\xi_k\|^2) \le C_2$.
\end{lemma}

Lemma~\ref{lemma:bounded_variance_2} is proved in Section~\ref{sec:proof_bounded_variance_2}. Then, by using Lemma~\ref{lemma:bounded_variance_2} and the Doob inequality as in relation~\eqref{eq:doob_inequality}, we demonstrated Assumption~\ref{ass:tadic}(b) of \cite{tadic2017asymptotic} (i.e. the noise fluctuation is controlled) and we can apply Theorem~\ref{theorem:tadic} to obtain equations~\eqref{eq1}-\eqref{eq2}.
Under Assumption~\ref{ass:g_bounded}, we study the asymptotic behavior of $\eta_k$,
\begin{align*}
    \|\eta_k\| &= \|\eE(\xi_k)\| = \frac{\lambda}{\sigma^2}\|\eE\left( \fJ_G^T(x_k) \left(D_{\sigma} - D^{\ast}_{\sigma}\right)(G(x_k)) \right)\| \\
    &\le \frac{\lambda}{\sigma^2} \eE\left( \vvvert \fJ_G(x_k) \vvvert \|\left(D_{\sigma} - D^{\ast}_{\sigma}\right)(G(x_k))\| \right).
\end{align*}
By Assumption~\ref{ass:g_bounded}, because $x_k \in \K$, we know that $G(x_k) \in \L = \mathcal{B}(0, C_{\K})$. So we have $\|\left(D_{\sigma} - D^{\ast}_{\sigma}\right)(G(x_k))\| \le \|D_{\sigma} - D^{\ast}_{\sigma}\|_{\infty, \L}$ and the desired inequality~\eqref{ineq}.
\end{proof}

\subsection{Critical points analysis - Geometrical invariant case}
A critical point analysis for the noising-denoising case, presented in equation~\eqref{eq:noising_denoising_denoiser}, is provided in~\cite{renaud2024plug} . 
However, to the best of our knowledge, no critical point analysis has been provided so far even in the case of geometrical invariance, i.e. $p \circ g = p$ for $g \in \G$. Here, we fill this gap by studying critical points under the relaxed $\pi$-equivariance property.
First, we study the approximation $- \nabla \log p \approx s_{\sigma}^{\pi}$ when $\sigma \to 0$. Next, we deduce that critical points of Problem~\eqref{eq:equivariant_opt_pblm} converge to critical points of Problem~\eqref{eq:ideal_opt_pblm} when $\sigma \to 0$.

\begin{assumption}
     \textbf{(a)} \label{ass:p_regularity} The prior distribution $p \in \mathrm{C}^1(\R^d, ]0, +\infty[)$ with $\|p\|_{\infty} + \|\nabla p\|_{\infty} < + \infty$.
\noindent
    \textbf{(b)} \label{ass:first_moment_finite}
    $\fJ_G$ has finite first moment, i.e. $\sup_{x \in \R^d}\eE_{G \sim \pi}(\vvvert \fJ_G(x)\vvvert) < +\infty$.
\end{assumption}
Assumption~\ref{ass:p_regularity}(a) is needed to ensure that $\nabla \log p$ is well defined.
Assumption~\ref{ass:first_moment_finite}(b) is verified in particular for a finite set of transformations, for a set of linear isometries or for the noising-denoising regularization.

\begin{proposition}\label{prop:uniform_cvg}
    Under Assumptions~\ref{ass:g_bounded}-\ref{ass:first_moment_finite}, for every compact $\K \subset \R^d$, if the prior $p$ is $\pi$-equivariant, we have, when $\sigma \to 0$,
    \begin{equation}
        \|s - s_{\sigma}^{\pi}\|_{\infty, \K} \to 0.
    \end{equation}
\end{proposition}

\begin{proof}
    Due to the $\pi$-equivariance of $p$, we have $s = \eE_{G \sim \pi}\left(\fJ_{G}^T (s \circ G)  \right)$ with a random variable $G \sim \pi$. With the definition of $s_{\sigma}^{\pi}$~\eqref{eq:equiv_score}, we get for $x \in \R^d$
\begin{align}
    &s_{\sigma}^{\pi}(x) - \nabla \log p(x) = \eE_{G \sim \pi} \left(\fJ_G^T(x) \left(\nabla \log p_{\sigma} - \nabla \log p \right)(G(x))\right).
\end{align}
By Assumption~\ref{ass:g_bounded}, we get that $\forall x \in \K, G(x) \in \mathcal{B}(0, C_{\K})$, with $\L = \mathcal{B}(0, C_{\K})$ the closed ball of center $0$ and radius $C_{\K}$. By Assumption~\ref{ass:p_regularity}(a) and Proposition 1 in~\cite{laumont2023maximum}, we know that $\|\nabla \log p_{\sigma} - \nabla \log p\|_{\infty, \L} \to 0$ when $\sigma \to 0$.

Then,  when $\sigma \to 0$, using Assumption~\ref{ass:g_bounded}, we obtain
\begin{align*}
    \|s - s_{\sigma}^{\pi}\|_{\infty, \K} &\le \eE_{G \sim \pi} \left(\vvvert\fJ_G\vvvert \|\nabla \log p_{\sigma} - \nabla \log p \|_{\infty, \L}\right) \\
    &\le \|\nabla \log p_{\sigma} - \nabla \log p \|_{\infty, \L} \eE_{G \sim \pi}\left(\vvvert\fJ_G\vvvert \right) \to 0. \qedhere
\end{align*}
\end{proof}

From Proposition~\ref{prop:uniform_cvg}, we now deduce  a  critical point convergence when $\sigma \to 0$ in the sense of Kuratowski~\cite{chambolle20231d}.

\begin{assumption}\label{ass:data_fidelity_diff}
    The data-fidelity term in~\eqref{eq:ideal_opt_pblm} is continuously differentiable, i.e. $f \in \mathcal{C}^1(\R^d, \R)$.
\end{assumption}
Assumption~\ref{ass:data_fidelity_diff} is needed to define the critical points of Problem~\eqref{eq:ideal_opt_pblm}. It is verified for a large set of inverse problems, including linear inverse problems with Gaussian noise, phase retrieval or despeckling.

We denote by $\mathbf{S}^\ast$ the set of critical point of $\fF$, 
$\mathbf{S}_{\sigma}$ the set of critical points of $\fF_{\sigma}^{\pi}$,  and $\mathbf{S}$ the limit point of $\mathbf{S}_{\sigma}$ when $\sigma \to 0$, more precisely
\begin{align}
    \mathbf{S} = \{ x \in \R^d | \exists \sigma_n>0 \text{ decreasing to } 0, x_n \in \mathbf{S}_{\sigma_n} \text{ such that } x_n \xrightarrow[n \to \infty]{} x\}.
\end{align}

\begin{proposition}\label{prop:critical_points_cvg}
    Under Assumptions~\ref{ass:g_bounded}-\ref{ass:data_fidelity_diff}, if the prior $p$ is $\pi$-equivariant, we have
    \begin{align}
        \mathbf{S} \subset \mathbf{S}^\ast.
    \end{align}
\end{proposition}

\begin{proof}
    For $x \in \S$, we have $\sigma_n >0$ decreasing to $0$ and $x_n \in \S_{\sigma_n}$ such that $x_n \to x$. Because $x_n$ is a converging sequence, there exists a compact $\K$ such that $\forall n \in \N, x_n \in \K$. Moreover, thanks to Proposition~\ref{prop:uniform_cvg}, $\|s - s_{\sigma_n}^{\G}\|_{\infty, \K} \xrightarrow[n \to \infty]{} 0$ and then $\|\fF - \fF_{\sigma_n}^{\pi}\|_{\infty, \K} \to 0$. So, $\|\fF(x_n) - \fF_{\sigma_n}^{\pi}(x_n)\| \to 0$ and by definition $\fF_{\sigma_n}^{\pi}(x_n) = 0$ which gives $\|\fF(x_n)\| \to 0$. It implies that $\|\nabla \fF(x)\| = 0$.
\end{proof}
\vspace*{-0.25cm}
\section{Experiments}

In this section, we evaluate the practical gain of Algorithm~\ref{alg:ERED} for image restoration. We focus on image deblurring with various blur kernels including fixed and motion kernels as proposed by~\cite{hurault2022gradient}. The denoiser used in these experiments is the Gradient-Step DRUNet denoiser proposed by~\cite{hurault2022gradient} with the provided weights, obtained with supervised training on natural color images. This denoiser reaches state-of-the-art denoising performance.

\begin{figure*}[!ht]
    \centering
    \includegraphics[width=\textwidth]{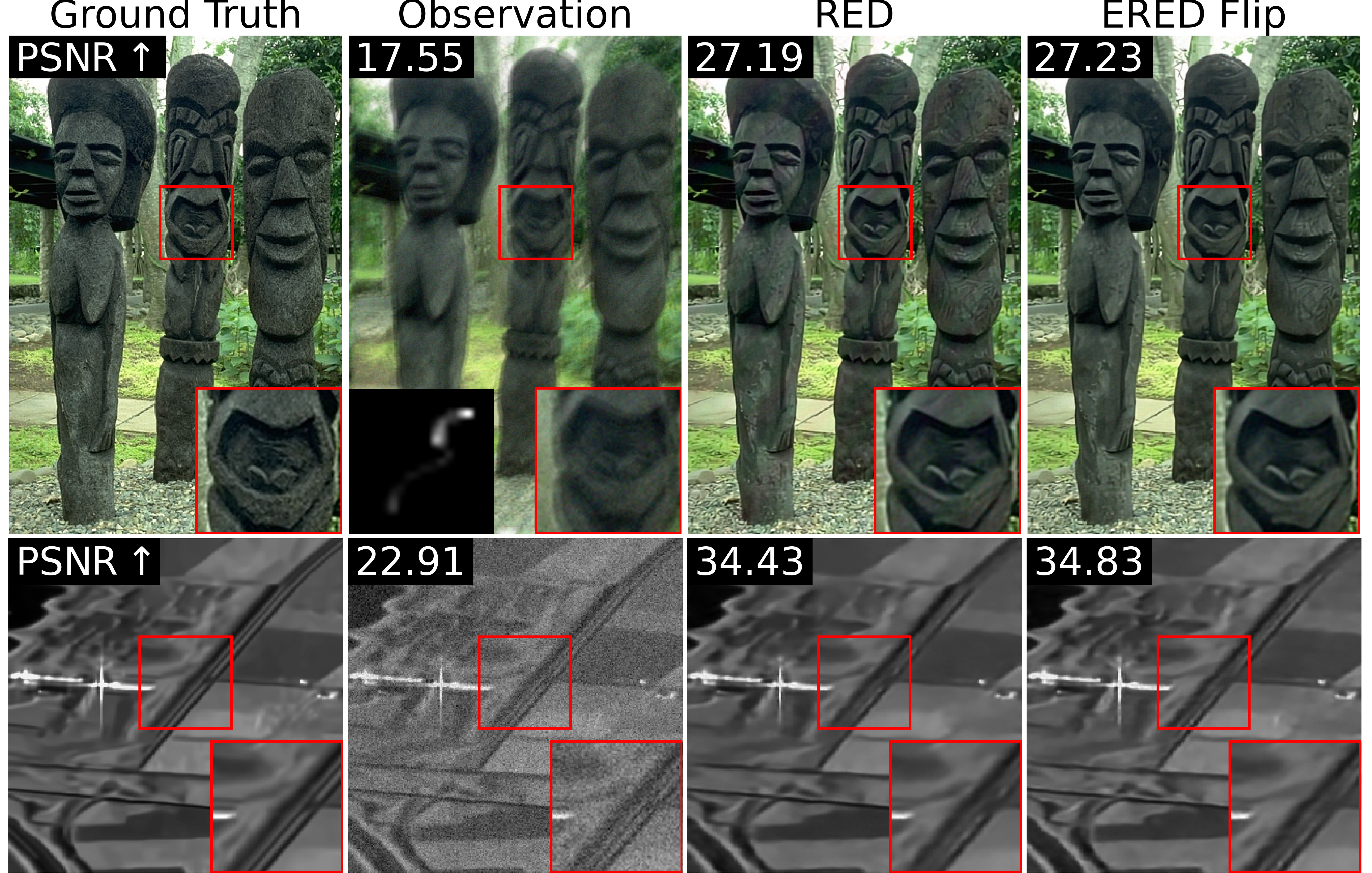}
    \caption{Deblurring (a motion blur kernel with input noise level $\sigma_{y} = 5 / 255$) and despeckling (number of looks $50$) with RED and ERED with a GS-denoiser trained on natural images or SAR images (respectively). The set of transformations for ERED is \textit{random flip}.  ERED produces a better qualitative result than RED.  \vspace*{-0.4cm}
    }
    \label{fig:restauration}
\end{figure*}

On Figure~\ref{fig:restauration}, we observe that ERED outperforms the standard RED algorithm. The set of transformations for ERED is \textit{random flip} of the two image axes. Especially, it succeeds to reduce artifacts generated by RED. In fact, ERED might impose more geometrical properties on the restored image due to its intrinsic equivariance.

\begin{wraptable}[13]{l}{0.47\textwidth}\label{table:quantitative_results}\vspace{-0.4cm}
\centering
\resizebox{0.9\linewidth}{!}{
\begin{tabular}{ c c c c }
 Method & PSNR$\uparrow$ & SSIM$\uparrow$ & N$\downarrow$\\
\hline
RED~\cite{romano2017little} & 32.25 & 0.84 & 400 \\
ERED rotation~\cite{terris2024equivariant} & \underline{32.53} & 0.85 & 400 \\
ERED translation & 32.44 & 0.85  & 400 \\
ERED flip & 32.51 & 0.85 & 400 \\
ERED subpixel rotation & 32.32 & 0.85 & 400\\
ERED all transformations & 31.94 & 0.83 & 400 \\
SNORE~\cite{renaud2024plug} & 32.45 & \underline{0.86} & 1000 \\
Annealed SNORE~\cite{renaud2024plug} & \textbf{32.89} & \textbf{0.87} & 1500\\
 \hline
\end{tabular}
} \vspace*{-0.25cm}
\caption{Quantitative comparison of image deblurring methods on $10$ images from CBSD68 dataset with $10$ different blur kernels (fixed and motion kernel of blur) and a noise level $\sigma_{y}  = 5 / 255$.
Best and second best results are respectively displayed in bold and underlined.}
\label{table:deblurring}
\end{wraptable}

On Table~\ref{table:quantitative_results}, we present Algorithm~\ref{alg:ERED} performance with various sets of transformations including random rotation of angle $\theta \in \{0,1,2,3\} \times \pi / 2$ named \textit{rotation}, random subpixel rotation of angle $\theta \in [-\pi, \pi]$ named \textit{subpixel rotation} using raster rotation~\cite{paeth1990fast}, random flip of the two axis named \textit{flip}, random translation along both axis named \textit{translation}, random Gaussian noising named \textit{SNORE}, and a random transformation taken randomly in between all the previous sets of transformations (including SNORE) named \textit{all transformations}. We also present performances of RED and Annealed SNORE, another version of SNORE~\cite{renaud2024plug} where the denoiser parameter $\sigma$ decreases through iterations. We notice that ERED  obtains better quantitative performance than RED ($+0.3$dB) with the same computational cost. Moreover, the choice of the transformation impacts the restoration quality and our experiments suggest that flips and rotations are beneficial.

Additional experiments are provided in Appendix~\ref{sec:additional_expe} on denoising, super-resolution, despeckling and deblurring with other denoisers.

\section{Conclusion}
In this paper, we propose ERED, an equivariant version of RED. We provide an interpretation of the ERED algorithm as an equivariant property of the underlying prior. We give theoretical convergence results (Propositions~\ref{prop:convergence_unbiased}-\ref{prop:convergence_biaised})  and a critical point convergence with an equivariant prior $p$ (Proposition~\ref{prop:critical_points_cvg}). Experimental results illustrate the modest improvement brought by such methods.

\section*{Acknowledgements}
This study has been carried out with financial support from the French Direction G\'en\'erale de l’Armement. Experiments presented in this paper were carried out using the PlaFRIM experimental testbed, 
supported by Inria, CNRS (LABRI and IMB), Universite de Bordeaux, Bordeaux INP and Conseil Regional d’Aquitaine (see https://www.plafrim.fr).

\bibliography{ref}
\bibliographystyle{abbrv}

\newpage
\appendix

\section{Additional experiments}\label{sec:additional_expe}
\paragraph{Denoising performances}
In order to understand the denoising benefit of $\pi$-equivariance, we study the denoising performance of the equivariant denoiser as defined in Equation~\eqref{eq:equiv_denoiser}. On table~\ref{table:denoising_performance}, we present the performance of each denoiser on natural images from the dataset CBSD68~\cite{martin2001} with various levels of noise. When the set of transformation is infinite, we take a Monte-Carlo approximation of Equation~\eqref{eq:equiv_denoiser} with $10$ random transformations. Note that the denoising performance are similar. Performance with subpixel rotation denoiser are slightly lower, which suggests that the image distribution might not be $\pi$-equivariant to subpixel rotation. Observing that the denoising results are similar with all approaches suggests that GS-DRUNet has already learned these equivariances. Therefore, the practical benefit of ERED may lie more in enforcing $\pi$-equivariance online rather than in introducing additional prior knowledge into the denoiser.

\begin{table}[]\label{table:denoising_performance}
\centering
\begin{tabular}{|c|c|c|c|}
\hline
Denoising method & PNSR,  & PNSR & PNSR\\
&$\sigma = 5 / 255$&$\sigma = 10 / 255$&$\sigma = 20 / 255$\\
\hline
Simple denoising & 40.54 & 36.46 & 32.73 \\
\hline
Rotation denoising & 40.58 & 36.49 & 32.76 \\
\hline
Translation denoising & 40.53 & 36.44 & 32.71\\ 
\hline
Subpixel Rotation denoising & 40.34 & 36.26 & 32.56\\
\hline
Flip denoising  & 40.58 & 36.49 & 32.76\\
\hline
\end{tabular}\vspace*{0.3cm}
\caption{Denoising results on the CBSD68 dataset with various level of noise. \textit{Simple denoising} refers to an application of the GS-DRUNet denoiser~\cite{hurault2022gradient}, \textit{rotation denoising} to the average of the denoising of the $4$ rotated images, \textit{flip denoising} to the average of the denoising of the $4$ flip images, \textit{translation denoising} to the average of the denoising of the $10$ random translated images and \textit{subpixel rotation denoising} to the average of the denoising of the $10$ random subpixel rotated images.}
\label{tab:my_label}
\end{table}

\paragraph{Super-resolution}
On Table~\ref{tab:sr} and Figure~\ref{fig:sr}, we present the super-resolution results with a super-resolution factor of $2$ for RED and ERED on natural images extracted form the CBSD68 dataset. As Table~\ref{table:deblurring} suggests that flips and rotations are the best $\pi$-equivariance, we focus on these transformations for this experiment. We considered $8$ blur kernels including motion and fixed kernels, with a noise level of $\sigma_{y} = 1/255$. We run $200$ iterations of each algorithm with a step-size of $\delta = 2.0$, a regularization parameter $\lambda = 0.07$ for RED and $\lambda = 0.05$ for ERED and a denoising parameter $\sigma = 11/255$ for RED and $\sigma = 13/255$ for ERED. A grid search has been made to find the optimal parameters in term on PSNR for each method.

On Table~\ref{tab:sr} on Figure~\ref{fig:sr}, we observe that the restoration performances are similar both quantitatively and qualitatively for RED and ERED. ERED does not seem to be useful for image super-resolution.

\begin{table}[]
\centering
\begin{tabular}{|c|c|c|c|}
\hline
Restoration method & PNSR $\uparrow$ & SSIM $\uparrow$ & N $\downarrow$\\
\hline
Bicubic & 25.47 & 0.72 & 200 \\
\hline
RED & 27.97 & 0.80 & 200\\
\hline
ERED Flip & 28.00 & 0.80 & 200\\
\hline
ERED Rotation & 28.01 & 0.80 & 200\\
\hline
\end{tabular}
\vspace*{0.2cm}
\caption{Super-resolution with a super-resolution factor of $2$ and $8$ different blur kernel (including fixed and motion blur) results on the CBSD10 ($10$ images from CBSD68) dataset with various restoration methods.\vspace*{-0.2cm}}
\label{tab:sr}
\end{table}

\begin{figure*}[!ht]
    \centering
    \includegraphics[width=\textwidth]{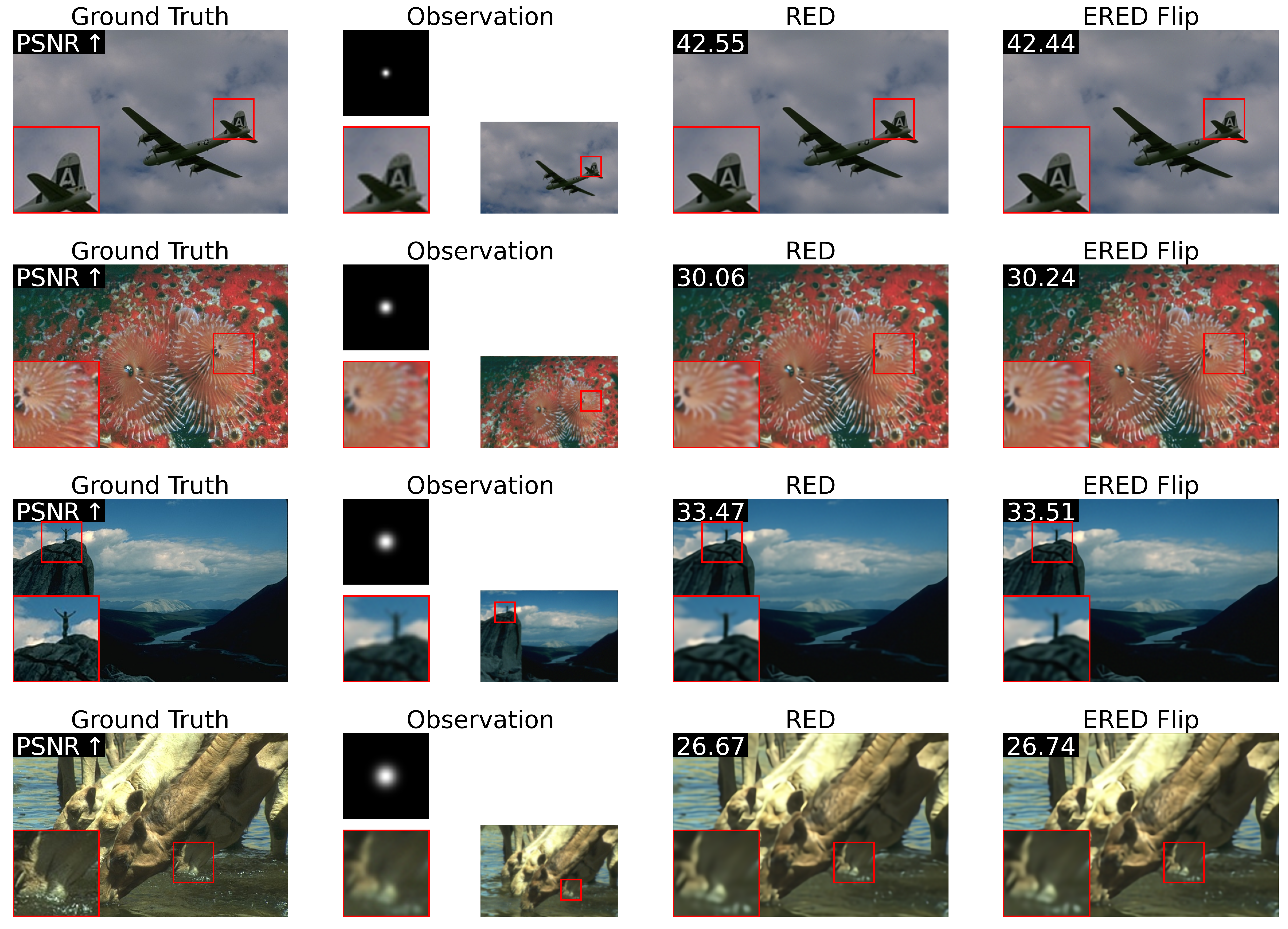}\vspace*{-0.2cm}
    \caption{Super-resolution with RED and ERED with a super-resolution factor of $2$ with a GS-denoiser trained on natural images. The set of transformation for ERED is random flip. Qualitative results of ERED and RED are very similar.  \vspace*{-0.2cm}
    }
    \label{fig:sr}
\end{figure*}

\paragraph{Despeckling}
On Table~\ref{tab:despeckle} and Figure~\ref{fig:despeckle}, we present the result of RED and ERED for Synthetic Aperture Radar (SAR) images despeckling. The speckle noised is multiplicative and implies a data-fidelity term that is not $L$-smooth. Therefore, it is known to be an inverse problem that is harder to tackle than Gaussian noise. We use for this experiment a dataset of SAR images presented in~\cite{dalsasso2021if}. The test image (\textit{lely}) of this dataset has been cropped into $60$ images of size $256 \times 256$ to create to our experimental dataset. The GS-Denoiser has been trained with the training images of this dataset with the parameters recommended in~\cite{hurault2022gradient}. Algorithms are run with $100$ iterations, a step-size $\delta = 0.01$, a denoiser parameter $\sigma = 8/255$, and a regularization parameter $\lambda = 100$. A grid search has been made to find the optimal parameters in term on PSNR for each method.

The PSNR values provided in Table~\ref{tab:despeckle} show that random rotations degrade restoration performances, whereas flips are beneficial.  This last observation is confirmed in Figure~\ref{fig:despeckle}, which illustrates that flips reduce the number of artifacts. Therefore, flip equivariance appears to be advantageous for despeckling.

\begin{table}[]
\centering
\begin{tabular}{|c|c|c|c|}
\hline
Restoration method & PNSR $\uparrow$ & SSIM $\uparrow$ & N $\downarrow$ \\
\hline
RED & 35.28 & 0.94 & 100\\
\hline
ERED Rotation & 35.19 & 0.94 & 100\\
\hline
ERED Flip & 35.59 & 0.94 & 100\\
\hline
\end{tabular}\vspace*{0.2cm}
\caption{Despeckle results on $60$ SAR images with various restoration methods. The number of looks is $L = 50$.}
\label{tab:despeckle}
\end{table}

\begin{figure*}[!ht]
    \centering
    \includegraphics[width=\textwidth]{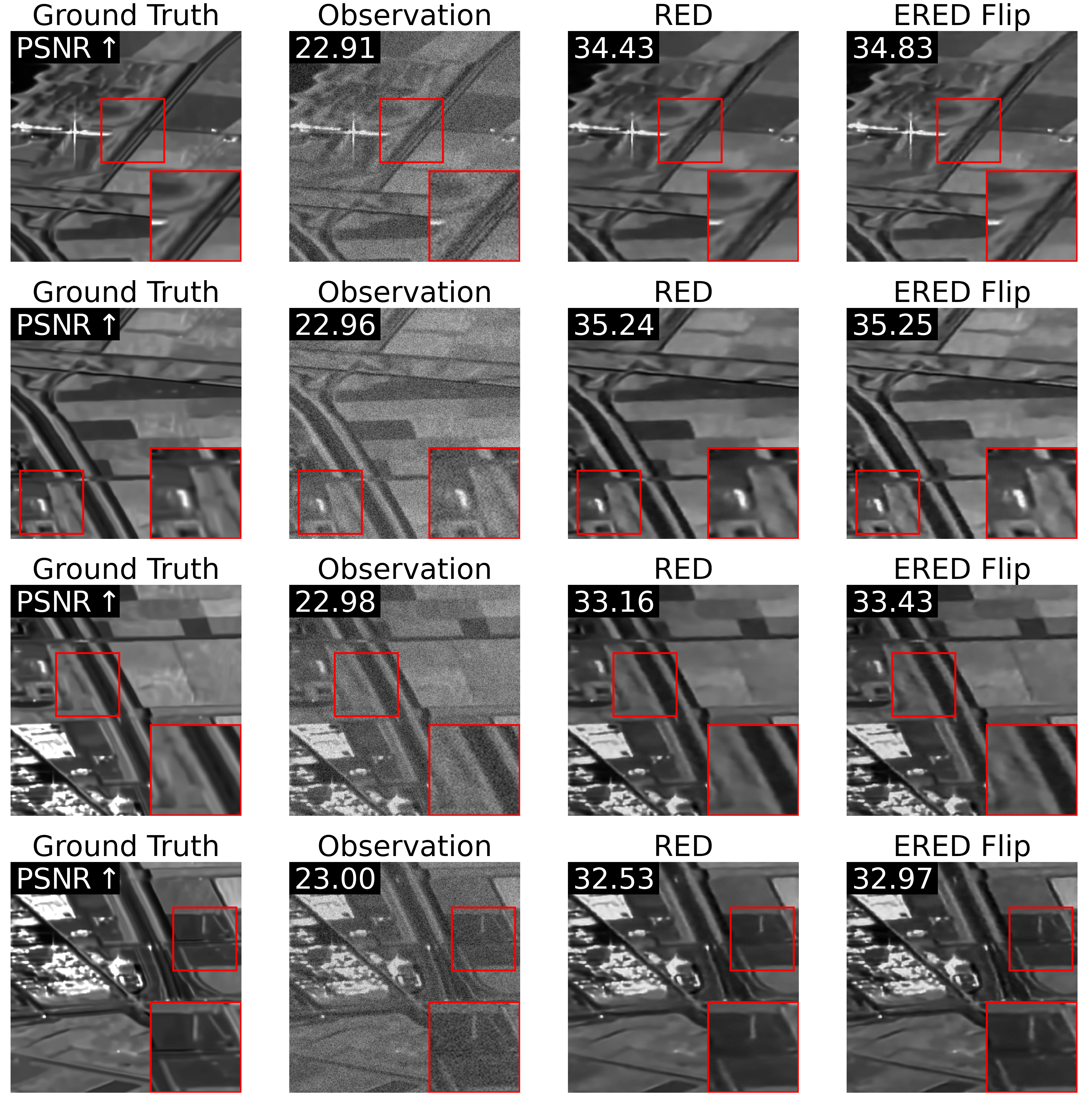}
    \caption{Despeckling with RED and ERED with a number of look of $L = 50$ with a GS-denoiser trained on SAR images. The set of transformation for ERED is random flip.  ERED produces a better qualitative result than RED.  \vspace*{-0.4cm}
    }
    \label{fig:despeckle}
\end{figure*}

\paragraph{Deblurring with different denoisers}
\textbf{GS-Denoiser} 
With the GS-Denoiser (Table~\ref{table:deblurring}, Table~\ref{tab:deblurring_deblurring_various_denoisers} and Figure~\ref{fig:restauration}) built on the DRUNet architecture~\cite{zhang2021plug}, we choose a step-size of $\delta = 1.5$, a regularization parameter $\lambda = 0.17$ for ERED, $\lambda = 0.15$ for RED, $\lambda = 0.5$ for SNORE, a denoiser parameter $\sigma = 8/255$ for ERED, $\sigma = 7/255$ for RED and $\sigma = 5/255$ for SNORE. A grid search has been made to find the optimal parameters in term of PSNR for each method. SNORE parameters have been chosen according to the recommendation in~\cite{renaud2024plug}. Experiments are made on deblurring with $\sigma_y = 5/255$ on $10$ various blur kernels including fixed and motion kernels.

With this denoiser the maximal gain of ERED compared to RED is $+0.3$dB.

\textbf{DRUNet}
On Table~\ref{tab:deblurring_deblurring_various_denoisers}, we also provide the quantitative results of RED and ERED (with flips and rotations) on natural images with the original DRUNet denoiser. We use a DRUNet with the pre-trained weights proposed the python librairie \textit{Deepinv}. This denoiser was trained with $\sigma \in [0,20]/255$. Experiments are made on deblurring with $\sigma_y = 5/255$ on $10$ various blur kernels including fixed and motion kernels. We use a step-size $\delta = 1.5$, a regularization weight $\lambda = 0.12$ for RED, $\lambda = 0.15$ for ERED with rotation and $\lambda = 0.14$ for ERED with flips. The denoiser parameter is $\sigma = 9/255$ for RED and $\sigma = 8/255$ for ERED. A grid search has been made to find the optimal parameters in term on PSNR for each method.

We observe a modest improvement of approximately $+0.2$ dB in equivariant methods compared to RED.

\textbf{DnCNN}
On Table~\ref{tab:deblurring_deblurring_various_denoisers}, we finally present the quantitative results of RED and ERED (with flips and rotations) on natural images. We use a DnCNN with the pre-trained weights shared by~\cite{pesquet2021learning}. This denoiser was trained with $\sigma = 2/255$. Therefore, experiments are made on deblurring with $\sigma_y = 1/255 < \sigma$ and $10$ various blur kernels including fixed and motion kernels. We use a step-size $\delta = 2.0$, a regularization weight $\lambda = 0.11$ and a denoiser parameter $\sigma = 2/255$ for every algorithm. A grid search has been made to find the optimal parameters in term on PSNR for each method.

We note a slight improvement of $+0.1$ dB for equivariant methods compared to RED.

These deblurring experiments with various denosiers suggest that the slight improvement brought by equivariance does not appear to be attributable to the training strategy or the denoiser architecture.
Hence contrary to the results provided in~\cite{terris2024equivariant}, we observed that the practical benefit of equivariance is not significant in our  numerical experiments.

\begin{table}[]
\centering
\begin{tabular}{|c|c|c|c|c|}
\hline
Denoiser & Restoration method & PNSR $\uparrow$ & SSIM $\uparrow$ & N $\downarrow$ \\
\hline
\multirow{3}{*}{GS-DRUNet ($\sigma_y = 5/255$)}  & RED & 32.25 & 0.84 & 400 \\
& ERED rotation & 32.53 & 0.85 & 400 \\
& ERED flip & 32.51 & 0.85 & 400 \\
\hline
\multirow{3}{*}{DRUNet ($\sigma_y = 5/255$)}  & RED & 29.24 & 0.81 & 400 \\
& ERED rotation & 29.48 & 0.83 & 400 \\
& ERED flip & 29.44 & 0.82 & 400 \\
\hline
\multirow{3}{*}{DnCNN ($\sigma_y = 1/255$)}  & RED & 35.26 & 0.94 &400 \\
& ERED rotation & 35.34 & 0.94 & 400 \\
& ERED flip & 35.32 & 0.94 &400 \\
\hline
\end{tabular}
\vspace*{0.2cm}
\caption{Deblurring results on CBSD10 ($10$ images extracted from CBSD68 dataset) with $10$ kernels of blur (including fixed and motion blur) with different pre-trained denoisers. It is worth noting that the quantitative improvement with equivariance is approximately $+0.2$ dB for each type of denoiser.}
\label{tab:deblurring_deblurring_various_denoisers}
\end{table}

\section{Technical proofs}
\subsection{Main technical result}\label{app:thm}
For completeness of the paper, here we recall Theorem 2.1 (ii) from~\cite{tadic2017asymptotic}. This theorem tackles the convergence of a sequence $x_n$ defined by a biased stochastic gradient descent algorithm, i.e. there exists $f:\R^d \to \R^d$ differentiable, such that 
$$x_{k+1} = x_k - \delta_k (\nabla f(x_k) + \xi_k), $$
with $\delta_k > 0$ the step-size and $\xi_k$ the bias noise.

\begin{assumption}\label{ass:tadic}
    \textbf{(a)} $\lim_{k\to +\infty}\delta_k = 0$ and $\sum_{k=0}^{+\infty} \delta_k = +\infty$. 
    
    \textbf{(b)} $\xi_k$ admits the decomposition $\xi_k = \zeta_k + \eta_k$ for all $k \ge 0$ that satisfy
    $$ \lim_{k\to +\infty} \max_{k\le n < a(k,t)} \| \sum_{i=k}^n \delta_k \zeta_k \| = 0, \limsup_{k \to +\infty}{\|\eta_k\|} < + \infty,$$
    almost surely on $\{\sup_{k \in \N}{\|x_k\|} < + \infty\}$. Where $a(k,t)$ is defined for $t> 0$ by $a(k,t) = \max \{  n \le k | \sum_{i=k}^{n-1} \delta_k \le t \}$.

    \textbf{(c)} $f$ is $p$-times differentiable on $\R^d$ with $p > d$.
\end{assumption}
    
\begin{theorem}{\cite{tadic2017asymptotic}}\label{theorem:tadic}
    Under Assumption~\ref{ass:tadic}, for a compact $Q \subset \R^d$ there exists a real number $K_{Q} > 0$ (depending only of $f$) such that it holds almost surely on $\lambda_Q = \{x_k \in Q | \forall k \in \N\}$  that
    \begin{equation}
        \limsup_{k \to + \infty} \|\nabla f(x_k)\| \le K \eta^{\frac{q}{2}},~~\limsup_{k \to + \infty} f(x_k) - \liminf_{k \to + \infty} f(x_k) \le K \eta^{q},
    \end{equation}
     with $q = \frac{p-d}{p-1}$ and $\eta = \limsup_{k \to +\infty} \|\eta_k\|$.
\end{theorem}

\subsection{Proof of Lemma~\ref{lemma:bounded_variance_1}}\label{sec:proof_bounded_variance_1}

By Assumption~\ref{ass:prior_score_approx}-\ref{ass:finite_moment}, and the inequality $\forall x,y \in \R^+, (x+y)^2 \le 2(x^2+y^2)$, we have
\begin{align*}
  &\eE(\|\xi_k\|^2 | x_k)\\
   =& \lambda^2 \eE(\|  \fJ_{G}^T(x_k) \nabla \log p_{\sigma}(G(x_k)) - \eE_{G \sim \pi}\left(  \fJ_{G}^T(x_k) \nabla \log p_{\sigma}(G(x_k)) \right)\|^2 | x_k) \\
   =& \lambda^2 \eE\left(\|  \fJ_{G}^T(x_k) \nabla \log p_{\sigma}(G(x_k))\|^2| x_k\right) - \lambda^2 \|\eE_{G \sim \pi}\left(  \fJ_{G}^T(x_k) \nabla \log p_{\sigma}(G(x_k)) | x_k \right)\|^2  \\
    \le &\lambda^2 \eE(\|  \fJ_{G}^T(x_k) \nabla \log p_{\sigma}(G(x_k))\|^2 | x_k) \\
    \le& \lambda^2 \eE(\vvvert  \fJ_{G}^T(x_k) \vvvert^2 \| \nabla \log p_{\sigma}(G(x_k))\|^2 | x_k) \\
    \le& \lambda^2 \eE(\vvvert  \fJ_{G}(x_k) \vvvert^2 B^2 \sigma^{2 \beta} \left( 1 + \|G(x_k)\|^{n_1} \right)^2 | x_k) \\
    \le& 2 \lambda^2 B^2 \sigma^{2 \beta} \eE(\vvvert  \fJ_{G}(x_k)\vvvert^2 \left( 1 + \|G(x_k)\|^{2n_1} \right) | x_k) \\
    \le& 2 \lambda^2 B^2 \sigma^{2 \beta} \left( \eE(\vvvert  \fJ_{G}(x_k)\vvvert^2| x_k) + \eE(\vvvert  \fJ_{G}(x_k)\vvvert^2 \|G(x_k)\|^{2n_1} | x_k)  \right).
\end{align*}
By using Young inequality, i.e. $\forall x, y \in\R^+, p, q > 1$ such that $\frac{1}{p}+ \frac{1}{q} = 1$, $|xy|\le \frac{x^p}{p} + \frac{y^q}{q}$, with $p= 1+\frac{\epsilon}{2}$, we get for $p = \frac{m_{n_1, \epsilon}}{\frac{2n_1(2+\epsilon)}{\epsilon}}$ with $m_{n_1, \epsilon} = \lceil \frac{2n_1(2+\epsilon)}{\epsilon} \rceil$
\begin{align*}
&\eE(\|\xi_k\|^2 | x_k) \\
    \le& 2 \lambda^2 B^2 \sigma^{2 \beta} \left( \frac{4}{2+\epsilon} \eE(\vvvert  \fJ_{G}(x_k)\vvvert^{2+\epsilon}| x_k) + \frac{\epsilon}{2+\epsilon}  + \frac{\epsilon}{2+\epsilon}\eE( \|G(x_k)\|^{\frac{2n_1(2+\epsilon)}{\epsilon}} | x_k)  \right)\\
    \le& 2 \lambda^2 B^2 \sigma^{2 \beta} \Bigl( \frac{4M_{2+\epsilon}+\epsilon}{2+\epsilon}   + \frac{1}{m_{n_1,\epsilon}} \Bigl( 2n_1 \eE( \|G(x_k)\|^{m_{n_1,\epsilon}} | x_k) + m_{n_1,\epsilon} \\&- \frac{2n_1(2+\epsilon)}{\epsilon} \Bigr) \Bigr)\\
    \le& 2 \lambda^2 B^2 \sigma^{2 \beta} \left( \frac{4M_{2+\epsilon}+\epsilon}{2+\epsilon}   + \frac{1}{m_{n_1,\epsilon}} \left( 2n_1 C_{\K, m_{n_1, \epsilon}} + m_{n_1,\epsilon} - \frac{2n_1(2+\epsilon)}{\epsilon} \right) \right) := C ,
\end{align*}
with $C< +\infty$ a constant independent of $k$ and $x_k$. This prove Lemma~\ref{lemma:bounded_variance_1} by taking the expectation on $x_k$ and the law of total expectation.

\subsection{Proof of Lemma~\ref{lemma:bounded_variance_2}}\label{sec:proof_bounded_variance_2}
Using the definition of $\xi_k$, Assumptions~\ref{ass:finite_moment}-\ref{ass:denoiser_sub_polynomial}, and Young inequality, we have, almost surely on $\Lambda_{\K}$,
\begin{align*}
    &\eE(\|\gamma_k\|^2|x_k) = \eE(\|\xi_k - \eE(\xi_k)\|^2|x_k) \\
    \le& \frac{\lambda^2}{\sigma^4} \eE\left( \| \fJ_{G}^T(x_k) \left( G(x_k) - D_{\sigma}(G(x_k))\right) \|^2|x_k \right) \\
    \le& \frac{\lambda^2}{\sigma^4} \eE\left( \vvvert\fJ_{G}(x_k) \vvvert^2 \|  G(x_k) - D_{\sigma}(G(x_k)) \|^2|x_k \right) \\
    \le &\frac{2\lambda^2}{\sigma^4} \eE\left( \vvvert\fJ_{G}(x_k) \vvvert^2 \left( \|G(x_k)\|^2 + \|D_{\sigma}(G(x_k)) \|^2 \right)|x_k \right) 
    \\\le &\frac{2\lambda^2}{\sigma^4} \eE\left( \vvvert\fJ_{G}(x_k) \vvvert^2 \left( \|G(x_k)\|^2 + 2 C^2 (1 + \|G(x_k) \|^{2n_2} \right)|x_k \right) \\
    \le& \frac{2\lambda^2}{\sigma^4} \Bigl( 2 C^2 \eE\left(\vvvert\fJ_{G}(x_k) \vvvert^2\right) + \eE\left(\vvvert\fJ_{G}(x_k) \vvvert^2 \|G(x_k)\|^2 \right)\\&+ 2 C^2 \eE\left(\vvvert\fJ_{G}(x_k) \vvvert^2 \|G(x_k)\|^{2n_2} \right)  \Bigr) \\
    \le &\frac{2\lambda^2}{\sigma^4} \bigg( \frac{4 C^2}{2+\epsilon} \eE\left(\vvvert\fJ_{G}(x_k)\vvvert^{2+\epsilon}\right) +\frac{\epsilon}{2+\epsilon} + \frac{2}{2+\epsilon} \eE\left(\vvvert\fJ_{G}(x_k) \vvvert^{2+\epsilon}\right)  \\
    & +\frac{\epsilon}{2+\epsilon}\eE\left( \|G(x_k)\|^{\frac{2(2+\epsilon)}{\epsilon}} \right)+ \frac{4 C^2}{2+\epsilon} \eE\left(\vvvert\fJ_{G}(x_k) \vvvert^{2+\epsilon} \right)\\& + \frac{2C^2\epsilon}{2+\epsilon} \eE\left( \|G(x_k)\|^{\frac{2n_2(2+\epsilon)}{\epsilon}} \right)\hspace{-0.1cm}\bigg) \\
    \le &\frac{2\lambda^2}{\sigma^4} \bigg( \frac{(8 C^2 + 2)M_{2+\epsilon}+\epsilon}{2+\epsilon}  +\frac{2}{\lceil c_{\epsilon} \rceil}\eE\left( \|G(x_k)\|^{\lceil c_{\epsilon} \rceil} \right) + \frac{\lceil c_{\epsilon} \rceil - c_{\epsilon}}{\lceil c_{\epsilon} \rceil} \\&+ \frac{4 n_2 C^2}{\lceil n_2 c_{\epsilon} \rceil} \eE\left( \|G(x_k)\|^{\lceil n_2 c_{\epsilon} \rceil} \right) + \frac{\lceil n_2 c_{\epsilon} \rceil - n_2 c_{\epsilon}}{\lceil n_2 c_{\epsilon} \rceil}  \bigg) \\
    \le &\frac{2\lambda^2}{\sigma^4} \hspace{-1pt}\bigg( \hspace{-1pt}\frac{(8 C^2 + 2)M_{2+\epsilon}+\epsilon}{2+\epsilon} \hspace{-1pt}  +\hspace{-1pt}\left(\frac{2}{\lceil c_{\epsilon} \rceil}\hspace{-1pt}+ \hspace{-1pt}\frac{4 n_2 C^2}{\lceil n_2 c_{\epsilon} \rceil}\right)\hspace{-1pt}C_{\K, \lceil c_{\epsilon} \rceil}\hspace{-1pt}+ \hspace{-1pt}\frac{\lceil c_{\epsilon} \rceil - c_{\epsilon}}{\lceil c_{\epsilon} \rceil} \hspace{-1pt}\\&+ \hspace{-1pt}\frac{\lceil n_2 c_{\epsilon} \rceil - n_2 c_{\epsilon}}{\lceil n_2 c_{\epsilon} \rceil}  \bigg) \\:= &C_2 < +\infty,
\end{align*}
with $c_{\epsilon} = \frac{2(2+\epsilon)}{\epsilon}$. This proves Lemma~\ref{lemma:bounded_variance_2} by taking the expectation on $x_k$.

\end{document}